\documentclass[12pt]{article}
\usepackage[left=35mm,right=35mm,top=35mm,bottom=35mm]{geometry}
\usepackage{amsmath}
\usepackage{amsthm}
\usepackage{amssymb}
\usepackage{amsfonts}
\usepackage{mathrsfs}

\DeclareMathOperator*{\slim}{s-lim}

\DeclareMathOperator{\supp}{supp}

\newcommand{\Eqn}[1]{&\hspace{-0.5em}#1\hspace{-0.5em}&}
\makeatletter

\@addtoreset{equation}{section}
\makeatother

\newtheorem{Prop}{Proposition}[section]
\newtheorem{Ass}[Prop]{Assumption}
\newtheorem{Thm}[Prop]{Theorem}
\newtheorem{Rem}[Prop]{Remark}
\newtheorem{Lem}[Prop]{Lemma}

\title{Threshold between short and long-range potentials for non-local Schr\"odinger operators}
\author{Atsuhide ISHIDA\\
\\
Department of Liberal Arts, Faculty of Engineering,
 \\Tokyo University of Science\\
\normalsize 6-3-1 Niijuku, Katsushika-ku,Tokyo 125-8585, Japan\\
\normalsize E-mail: aishida@rs.tus.ac.jp\\
\normalsize Fax: +81-3-5876-1616
}
\date{}

\begin{document}
\begin{flushleft}
{\Large \bf Threshold between short and long-range potentials for non-local Schr\"odinger operators}
\end{flushleft}

\begin{flushleft}
{\large Atsuhide ISHIDA}\\
{Department of Liberal Arts, Faculty of Engineering, Tokyo University of Science, 6-3-1 Niijuku, Katsushika-ku,Tokyo 125-8585, Japan\\ 
Email: aishida@rs.tus.ac.jp
}
\end{flushleft}
\begin{flushleft}
{\large Kazuyuki WADA}\\
{Department of General Science and Education, National Institute of Technology, Hachinohe College, 16-1 Uwanotai Tamonoki, Hachinohe 039-1192, Japan}\\
Email: {wada-g@hachinohe.kosen-ac.jp}
\end{flushleft}

\begin{abstract}
We develop scattering theory for non-local Schr\"odinger operators defined by functions of the Laplacian that include its fractional power $(-\Delta)^\rho$ with $0<\rho\leqslant1$. In particular, our function belongs to a wider class than the set of Bernstein functions. By showing the existence and non-existence of the wave operators, we clarify the threshold between the short and long-range decay conditions for perturbational potentials.
\end{abstract}

\quad\textit{Keywords}: Bernstein function, scattering theory, wave operators\par
\quad\textit{MSC}2010: 47B25, 81Q10, 81U05

\section{Introduction\label{introduction}}
We consider the set of functions
\begin{equation}
\tilde{\mathscr{B}}=\left\{\Psi\in C^1\left(0,\infty\right)\bigm|\Psi(\sigma)\geqslant0, \frac{d\Psi}{d\sigma}(\sigma)\geqslant0, \sigma\in(0,\infty)\right\}.
\end{equation}
For $\Psi\in\tilde{\mathscr{B}}$, the free Hamiltonian we consider in this paper is given by
\begin{equation}
H_0^{\Psi}=\Psi(-\Delta)=\Psi\left(|D|^2\right),
\end{equation}
a self-adjoint operator acting on $L^2(\mathbb{R}^n)$, where $D=-{\rm i}\left(\partial_{x_1},\ldots,\partial_{x_n}\right)$. Let $V$ be a real-valued function satisfying suitable conditions. We call operator $H_0^{\Psi}+V$ a non-local Schr\"odinger operator.\par
An element of the following central set of functions
\begin{equation}
\mathscr{B}=\left\{\Phi\in C^\infty\left(0,\infty\right)\bigm|\Phi(\sigma)\geqslant0, (-1)^k\frac{d^k\Phi}{d\sigma^k}(\sigma)\leqslant0, \sigma\in(0,\infty), k\in\mathbb{N} \right\}
\end{equation}
is called a Bernstein function. Clearly, the relation
\begin{equation}
\mathscr{B}\subsetneq\tilde{\mathscr{B}}
\end{equation}
holds. Therefore, we treat much more general operators than those defined as Bernstein functions.\par
Consider $\Phi\in\mathscr{B}$ that satisfy $\lim_{\sigma\rightarrow+0}\Phi(\sigma)=0$. One important property of $\Phi$ is that the semigroup generated by $\Phi\left(|D|^2\right)+V$ is expressible via a stochastic process
\begin{equation}
\left(e^{-t\left(\Phi\left(|D|^2\right)+V\right)}f\right)(x)=\mathbb{E}^x\left[e^{-\int_0^tV\left(B_{T(s)}\right)ds}f\left(B_{T(t)}\right)\right]
\end{equation}
where $\{B_{t}\}_{t\geqslant0}$ is the $n$-dimensional Brownian motion starting at $x\in\mathbb{R}^n$ and $T$ is the subordinator associated with $\Phi$. The above expression is often called Feynman-Kac formula or path integral representation.\par
The Feynman--Kac formula enables us to analyze $\Phi\left(|D|^2\right)+V$ in terms of stochastic calculus. In particular, it is useful to investigate some properties of the eigenfunctions of $\Phi\left(|D|^2\right)+V$. In this direction, there already exist several results \cite{HiLo1,HiIcLo,KaLo}. Comprehensive results related to the Feynman--Kac formula are summarized in \cite{LoHiBe}.\par
With regard to quantum scattering theory for non-local Schr\"odinger operators, there are only a few results. For instance, \cite{Gi} considered a function $\Phi\in\hat{\mathscr{B}}$, where
\begin{equation}
\hat{\mathscr{B}}=\left\{\Phi\in C^\infty\left(0,\infty\right)\bigm|\frac{d\Phi}{d\sigma}(\sigma)\geqslant0, \sigma\in(0,\infty), \lim_{\sigma\rightarrow\infty}\Phi(\sigma)=\infty\right\}.
\end{equation}
\cite{Gi} discusses the asymptotic completeness of the wave operators under a short-range perturbation $V$ by investigating the semigroup differences and proves the absence of the singular continuous spectrum of $\Phi\left(|D|^2\right)+V$. Of note, instances of fractional powers are specific examples of functions of the Laplacian. \cite{Ki2, Ki3} first constructed long-range scattering theory for $(-\Delta)^\rho$ with $1/2\leqslant\rho\leqslant1$. Recently, inverse scattering problems involving the short-range potentials were investigated in \cite{Is3} for exponents $\rho$ satisfying $1/2<\rho\leqslant1$.\par
Under these situations, our aim in this paper is to develop scattering theory for non-local Schr\"{o}dinger operators defined by functions belonging to the wider set $\tilde{\mathscr{B}}$ with some additional assumptions. We first prove the existence of the wave operators in which the potential function $V$ has a short-range decay. \cite{Gi} also proved these existence for $\Phi\in\hat{\mathscr{B}}$ by way of semigroups $e^{-t\Phi\left(|D|^2\right)}$ and $e^{-t\left(\Phi\left(|D|^2\right)+V\right)}$. However, our proof of the existence of the wave operators is simple and intuitive, and is obtained directly from the propagation estimate for the time-evolution of $e^{-{\rm i}tH_0^\Psi}$. We next clarify the threshold between the short and long-range decay exponents by providing a concrete example of the potential functions for which the wave operators do not exist.\par
For quantum scattering theory, it is important to distinguish the threshold between the short and long-range conditions. Historically, \cite{Do} considered the non-existence of the standard wave operators for the Coulomb interaction by showing that the weak limits of the pair of propagators were equal to zero (see also \cite{ReSi}). \cite{Is2, Is4} applied this method to the fractional Laplacian and massive relativistic operator. From a different perspective, \cite{Oz} invented an original approach to prove the non-existence of the wave operators for the Stark Hamiltonian. This approach was applied to various quantum systems (repulsive Hamiltonian \cite{Is1}, time-dependent harmonic oscillator \cite{IsKa} and 1D quantum walk \cite{Wa}). Our approach in this paper follows \cite{Oz}.\\\par

First of all, we provide the basic properties associated with the spectrum of $H_0^{\Psi}$. The absolute continuity of $\Phi\left(|D|^2\right)$ for $\Phi\in\hat{\mathscr{B}}$ was mentioned in \cite{Gi} (Remark 2.2). We also state several other properties of $H_0^{\Psi}=\Psi\left(|D|^2\right)$ for $\Psi\in\tilde{\mathscr{B}}$ in the following proposition.

\begin{Prop}\label{Prop0}
\begin{enumerate}
Assume that $\Psi\in\tilde{\mathscr{B}}$. 
\item
The spectrum of $H_0^{\Psi}$ coincides with
\begin{equation}
\left[\lim_{\sigma\rightarrow+0}\Psi(\sigma),\lim_{\sigma\rightarrow\infty}\Psi(\sigma)\right]\label{prop0_1}
\end{equation}
if $\lim_{\sigma\rightarrow\infty}\Psi(\sigma)<\infty$,
\begin{equation}
\left[\lim_{\sigma\rightarrow+0}\Psi(\sigma),\infty\right)\label{prop0_2}
\end{equation}
if $\lim_{\sigma\rightarrow\infty}\Psi(\sigma)=\infty$.

\item
If the set
\begin{equation}
A=\left\{\sigma\in(0,\infty)\bigm|\Psi'(\sigma)=0\right\} 
\end{equation}
is at most discrete, the pure point spectrum of $H_0^{\Psi}$ is empty. If there exists a proper interval $I\subset A$, for $\sigma\in I$ fixed, $\Psi(\sigma)$ is an eigenvalue of $H_0^{\Psi}$ with infinite multiplicity.

\item
If the set $A$ is at most discrete, $H_0^{\Psi}$ is absolutely continuous.
\end{enumerate}
\end{Prop}

\begin{proof}
\begin{enumerate}
\item
By the usual Fourier transform on $L^2(\mathbb{R}^n)$, we represent
\begin{equation}
H_0^{\Psi}=\mathscr{F}^*\Psi\left(|\xi|^2\right)\mathscr{F}.
\end{equation}
Therefore, its spectrum is given by the closure of the range of $\Psi\left(|\xi|^2\right)$. Because $\Psi$ is continuous and monotonically increasing, the spectrum of $H_0^{\Psi}$ coincides with \eqref{prop0_1} or \eqref{prop0_2}. 

\item
If $A$ is a discrete set, then because $\Psi$ is continuous and monotonic, for $\lambda\geqslant\lim_{\sigma\rightarrow+0}\Psi(\sigma)$, there is only one $\sigma_\lambda\in(0,\infty)$ such that $\Psi(\sigma_\lambda)=\lambda$ and the $n$-dimensional Lebesgue measure of
\begin{equation}
\left\{\xi\in\mathbb{R}^n\bigm||\xi|^2=\sigma_\lambda\right\}
\end{equation}
is zero. Therefore, if we assume that there exists $u\in L^2\left(\mathbb{R}^n\right)$ such that 
\begin{equation}
\left(\Psi\left(|\xi|^2\right)-\lambda\right)u(\xi)=0,
\end{equation}
then $u=0$ holds. This implies that the pure point spectrum of $H_0^{\Psi}$ is empty. Alternatively, if $I\subset A$ is a proper interval, for $\sigma\in I$, the Lebesgue measure of
\begin{equation}
\left\{\xi\in\mathbb{R}^n\bigm|\Psi\left(|\xi|^2\right)=\Psi(\sigma)\right\}\label{prop0_3}
\end{equation}
is positive. In this case, $0\not=u\in L^2\left(\mathbb{R}^n\right)$, which has support in \eqref{prop0_3}, satisfies
\begin{equation}
\left(\Psi\left(|\xi|^2\right)-\Psi(\sigma)\right)u(\xi)=0.
\end{equation}
This implies that $\Psi(\sigma)$ is an eigenvalue of $H_0^{\Psi}$ and $u$ is the corresponding eigenfunction. In particular, \eqref{prop0_3} has infinite disjoint subsets for which the Lebesgue measures are positive. This also implies that $\Psi(\sigma)$ has infinite multiplicity. 

\item
If the one-dimensional Lebesgue measure of the Borel set $B$ is equal to zero, the $n$-dimensional Lebesgue measure of
\begin{equation}
\left\{\xi\in\mathbb{R}^n\bigm|\Psi\left(|\xi|^2\right)\in B\right\}
\end{equation}
is also zero because $A$ is discrete. Therefore,
\begin{equation}
\int_{\Psi\left(|\xi|^2\right)\in B}\left|\left(\mathscr{F}u\right)(\xi)\right|^2d\xi=0
\end{equation}
holds for $u\in L^2(\mathbb{R}^n)$. This shows the absolute continuity of $H_0^{\Psi}$.
\end{enumerate}
\end{proof}

\begin{Rem}\label{Rem1}
Statements 1 and 2 in Proposition \ref{Prop0} can be replaced by the following. If the Lebesgue measure of $A$ is zero, the pure point spectrum of $H_0^{\Psi}$ is empty and $H_0^{\Psi}$ is absolutely continuous. If the Lebesgue measure of $A$ is positive, $H_0^{\Psi}$ has an eigenvalue with infinite multiplicity. These proofs are demonstrated in a similar manner to that above.
\end{Rem}

\begin{Ass}\label{Ass1}
Let $\Psi_\pm\in\tilde{\mathscr{B}}$ be fixed and suppose $\Psi'_\pm=d\Psi_\pm/d\sigma>0$. In addition, we assume that $\Psi_+'(\sigma^2)\sigma$ increases monotonically and $\Psi_-'(\sigma^2)\sigma$ decreases monotonically, for $\sigma\in(0,\infty)$.
\end{Ass}

\begin{Rem}\label{Rem2}
Under this assumption, $H_0^{\Psi_\pm}$ do not have any eigenvalues and are absolutely continuous because $\Psi'_\pm>0$.
\end{Rem}

The monotonicity in Assumption \ref{Ass1} is not extraordinary and it is not difficult to remove this assumption (see Remark \ref{Rem3} immediately following Proposition \ref{Prop1}). For example, $\Psi_+(\sigma)=\sqrt{\sigma}$ is allowed although $\Psi_+'(\sigma^2)\sigma$ is always equal to $1/2$. In this case, as is well known,
\begin{equation}
\Psi_+\left(|D|^2\right)=\sqrt{-\Delta}
\end{equation}
is the massless relativistic Schr\"odinger operator. More generally, Assumption \ref{Ass1} admits fractional Schr\"odinger operators $\Psi_+\left(|D|^2\right)=(-\Delta)^\rho$ with $1/2\leqslant\rho\leqslant1$ and $\Psi_-\left(|D|^2\right)=(-\Delta)^\rho$ with $0<\rho<1/2$. We assume that $\Psi'_\pm(\sigma^2)\sigma$ is monotonic when keeping in mind fractional exponents.\par
In the monotonically increasing case, high-energy is expressed as $\Psi'_+\left(|\xi|^2\right)|\xi|$ with $|\xi|\gg1$, whereas in monotonically decreasing case, high-energy bocomes $\Psi'_-\left(|\xi|^2\right)|\xi|$ with $|\xi|\ll1$. This reverse trend in the momentum space is interesting and is described by the difference of two inequalities \eqref{prop1_1} and \eqref{prop1_2} in Proposition \ref{Prop1}.

\begin{Ass}\label{Ass2}
Let $V^{\rm S}\in L^\infty(\mathbb{R}^n)$. There exist positive constant $C$ and exponents $\gamma_{\rm S}>1$ such that
\begin{equation}
|V^{\rm S}(x)|\leqslant C\langle x\rangle^{-\gamma_{\rm S}},\label{short_range}
\end{equation}
where $\langle x\rangle=\sqrt{1+|x|^2}$. For $0\not=\kappa\in\mathbb{R}$ and $0<\gamma_L\leqslant1$, we also define $V^{\rm L}\in L^\infty(\mathbb{R}^n)$ by
\begin{equation}
V^{\rm L}(x)=\kappa\langle x\rangle^{-\gamma_L}.\label{long_range}
\end{equation}
\end{Ass}
Throughout this paper, $\phi\in L^2(\mathbb{R}^n)$ satisfies $\mathscr{F}\phi\in C_0^\infty(\mathbb{R}^n\setminus\{0\})$. In particular, for $R>\epsilon>0$ fixed, we assume that 
\begin{equation}
\supp\mathscr{F}\phi\subset\{ \xi\in\mathbb{R}^n\bigm|\epsilon\leqslant|\xi|\leqslant R\}.
\end{equation}
$F(\cdots)$ denotes the characteristic function of the set $\{\cdots\}$. Moreover, we write the full Hamiltonians in the form
\begin{equation}
H_{\rm S}^{\Psi_\pm}=H_0^{\Psi_\pm}+V^{\rm S},\quad H_{\rm L}^{\Psi_\pm}=H_0^{\Psi_\pm}+V^{\rm L}.
\end{equation}

\begin{Thm}\label{Thm1}
The wave operators
\begin{equation}
\slim_{t\rightarrow\infty}e^{{\rm i}tH_{\rm S}^{\Psi_\pm}}e^{-{\rm i}tH_0^{\Psi_\pm}},\quad\slim_{t\rightarrow-\infty}e^{{\rm i}tH_{\rm S}^{\Psi_\pm}}e^{-{\rm i}tH_0^{\Psi_\pm}}\label{existence}
\end{equation}
exist. However, 
\begin{equation}
\slim_{t\rightarrow\infty}e^{{\rm i}tH_{\rm L}^{\Psi_\pm}}e^{-{\rm i}tH_0^{\Psi_\pm}},\quad\slim_{t\rightarrow-\infty}e^{{\rm i}tH_{\rm L}^{\Psi_\pm}}e^{-{\rm i}tH_0^{\Psi_\pm}}\label{nonexistence}
\end{equation}
do not exist. This means that the threshold between short and long-range depends on whether the decay exponent of the potential function is less than $-1$, or greater than or equal to $-1$
\end{Thm}

To prove Theorem \ref{Thm1}, several Propositions and Lemmas are needed. In the following, we only consider the limit $t\rightarrow\infty$ because the other case is proved similarly.

\section{Existence of Wave Operators}

In this section, we prove the existence of the wave operators. Although the following propagation estimates for free evolution $e^{-{\rm i}tH_0^{\Psi_\pm}}$ are simple, these estimates also work well in the next section. 
\begin{Prop}\label{Prop1}
Let $t>0$ and $N\in\mathbb{N}$. There exist positive constants $C_{\pm,N,\epsilon, R}$ such that
\begin{equation}
\left\| F\left(\frac{|x|}{t}\leqslant\Psi'_+(\epsilon^2)\epsilon\right)e^{-{\rm i}tH_0^{\Psi_+}}\phi\right\|_{L^2(\mathbb{R}^n)}\leqslant C_{+,N,\epsilon, R}t^{-N}\left\|\langle x\rangle^N\phi\right\|_{L^2(\mathbb{R}^n)}\label{prop1_1}
\end{equation}
and
\begin{equation}
\left\| F\left(\frac{|x|}{t}\leqslant\Psi'_-(R^2)R\right)e^{-{\rm i}tH_0^{\Psi_-}}\phi\right\|_{L^2(\mathbb{R}^n)}\leqslant C_{-,N,\epsilon, R}t^{-N}\left\|\langle x\rangle^N\phi\right\|_{L^2(\mathbb{R}^n)}\label{prop1_2}
\end{equation}
hold. 
\end{Prop}

\begin{proof}
There exists a function $f\in C_0^\infty(\mathbb{R}^n\setminus\{0\})$ with $\supp f\subset\{\epsilon\leqslant|\xi|\leqslant R\}$ such that $\phi=f(D)\phi$. We then find
\begin{eqnarray}
\lefteqn{F\left(\frac{|x|}{t}\leqslant\Psi'_+(\epsilon^2)\epsilon\right)e^{-{\rm i}tH_{\rm S}^{\Psi_+}}\phi}\nonumber\\
\Eqn{}=\frac{1}{(2\pi)^{n/2}}\int_{\mathbb{R}^n_\xi}e^{{\rm i}\left(x\cdot\xi-t\Psi_+\left(|\xi|^2\right)\right)}F\left(\frac{|x|}{t}\leqslant\Psi'_+(\epsilon^2)\epsilon\right)f(\xi)\left(\mathscr{F}\phi\right)(\xi)d\xi.
\end{eqnarray}
When $|\xi|\geqslant\epsilon$ and $|x|/t\leqslant\Psi'_+(\epsilon^2)\epsilon$ hold, we see that
\begin{equation}
\left|\nabla_\xi\left(x\cdot\xi-t\Psi_+\left(|\xi|^2\right)\right)\right|\geqslant2t\Psi'_+\left(|\xi|^2\right)|\xi|-|x|\geqslant t\Psi'_+(\epsilon^2)\epsilon
\end{equation}
because $\Psi'_+(\sigma^2)\sigma$ is monotonically increasing for $\sigma>0$. Using this inequality and the relation
\begin{equation}
-i\frac{\nabla_\xi\left(x\cdot\xi-t\Psi_+\left(|\xi|^2\right)\right)\cdot\nabla_\xi}{\left|\nabla_\xi\left(x\cdot\xi-t\Psi_+\left(|\xi|^2\right)\right)\right|^2}e^{{\rm i}\left(x\cdot\xi-t\Psi_+\left(|\xi|^2\right)\right)}=e^{{\rm i}\left(x\cdot\xi-t\Psi_+\left(|\xi|^2\right)\right)},
\end{equation}
\eqref{prop1_1} follows from the standard integration by parts method (see Kitada \cite{Ki1} for example). As for \eqref{prop1_2}, when $|\xi|\leqslant R$ and $|x|/t\leqslant\Psi'_-(R^2)R$ hold, we see that
\begin{equation}
|\nabla_\xi(x\cdot\xi-t\Psi_-\left(|\xi|^2\right))|\geqslant t\Psi'_-(R^2)R.
\end{equation}
We therefore also obtain \eqref{prop1_2}.
\end{proof}

\begin{Rem}\label{Rem3}
If we do not assume the monotonicity of $\Psi'_\pm(\sigma^2)\sigma$, the estimates \eqref{prop1_1} and \eqref{prop1_2} are replaced by
\begin{equation}
\left\| F\left(\frac{|x|}{t}\leqslant\inf_{\epsilon\leqslant\sigma\leqslant R}\Psi'(\sigma^2)\sigma\right)e^{-{\rm i}tH_{\rm S}^\Psi}\phi\right\|_{L^2(\mathbb{R}^n)}\leqslant C_{N,\epsilon, R}t^{-N}\left\|\langle x\rangle^N\phi\right\|_{L^2(\mathbb{R}^n)},
\end{equation}
for $\Psi\in\tilde{\mathscr{B}}$ which satisfies $\Psi'>0$. The following proofs also proceed without monotonicity. 
\end{Rem}

Proposition \ref{Prop1} yields the existence of the wave operators immediately.

\begin{proof}[Proof of the existence of the wave operators]
Let us first consider the existence of
\begin{equation}
\slim_{t\rightarrow\infty}e^{{\rm i}tH_{\rm S}^{\Psi_+}}e^{-{\rm i}tH_0^{\Psi_+}}.\label{wave_operators_+}
\end{equation}
The derivative at $t$ of $e^{{\rm i}tH_{\rm S}^{\Psi_+}}e^{-{\rm i}tH_0^{\Psi_+}}\phi$ is
\begin{eqnarray}
\lefteqn{\frac{d}{dt}e^{{\rm i}tH_{\rm S}^{\Psi_+}}e^{-{\rm i}tH_0^{\Psi_+}}\phi={\rm i}e^{{\rm i}tH_{\rm S}^{\Psi_+}}V^{\rm S}e^{-{\rm i}tH_0^{\Psi_+}}\phi}\nonumber\\
\Eqn{}={\rm i}e^{{\rm i}tH_{\rm S}^{\Psi_+}}F\left(\frac{|x|}{t}>\Psi'_+(\epsilon^2)\epsilon\right)V^{\rm S}e^{-{\rm i}tH_0^{\Psi_+}}\phi\nonumber\\
\Eqn{}\quad+{\rm i}e^{{\rm i}tH_{\rm S}^{\Psi_+}}F\left(\frac{|x|}{t}\leqslant\Psi'_+(\epsilon^2)\epsilon\right)V^{\rm S}e^{-{\rm i}tH_0^{\Psi_+}}\phi.\label{thm1_1}
\end{eqnarray}
We abbreviate the $L^2$-norm $\|\cdot\|_{L^2(\mathbb{R}^n)}$ to $\|\cdot\|$ for simplicity and estimate
\begin{equation}
\left\|\frac{d}{dt}e^{{\rm i}tH_{\rm S}^{\Psi_+}}e^{-{\rm i}tH_0^{\Psi_+}}\phi\right\|\leqslant C\langle t\Psi'_+(\epsilon^2)\epsilon\rangle^{-\gamma_{\rm S}}\|\phi\|+\left\|V^{\rm S}\right\|_{L^\infty(\mathbb{R}^n)}C_{+,N,\epsilon, R}t^{-N}\left\|\langle x\rangle^N\phi\right\|,\label{thm1_2}
\end{equation}
where we have used the decay assumption \eqref{short_range} and Proposition \ref{Prop1}. With $\gamma_{\rm S}>1$, we can choose $N\in\mathbb{N}$ as $N\geqslant2$. Then \eqref{thm1_2} implies the existence of \eqref{wave_operators_+} by the Cook--Kuroda method and a density argument. The existence of
\begin{equation}
\slim_{t\rightarrow\infty}e^{{\rm i}tH_{\rm S}^{\Psi_-}}e^{-{\rm i}tH_0^{\Psi_-}}.\label{wave_operators_-}
\end{equation}
is proved by simply replacing $\Psi'_+(\epsilon^2)\epsilon$ with $\Psi'_+(R^2)R$ inside the characteristic functions of \eqref{thm1_1}.
\end{proof}

\section{Non-existence of Wave Operators}

This section is devoted to proving the non-existence of the wave operators when the potential function $V^{\rm L}$ satisfies \eqref{long_range}.

\begin{Lem}\label{Lem1}
For $t\geqslant1$, there exist positive constants $c_{\pm,1}$ and $c_{\pm,2}$ such that
\begin{equation}
\frac{1}{\kappa}\left(V^{\rm L}e^{-{\rm i}tH_0^{\Psi_\pm}}\phi,e^{-{\rm i}tH_0^{\Psi_\pm}}\phi\right)_{L^2(\mathbb{R}^n)}\geqslant c_{\pm,1}t^{-\gamma_{\rm L}}\|\phi\|_{L^2(\mathbb{R}^n)}^2-c_{\pm,2}t^{-2-\gamma_{\rm L}}\left\|x\phi\right\|_{L^2(\mathbb{R}^n)}^2,\label{Lem1_1}
\end{equation}
where $(\cdot,\cdot)_{L^2(\mathbb{R}^n)}$ denotes the scalar product on $L^2(\mathbb{R}^n)$.
\end{Lem}

\begin{proof}
By a straightforward computation, we see that the Heisenberg representation of the position $x$ is
\begin{equation}
e^{{\rm i}tH_0^{\Psi_\pm}}xe^{-{\rm i}tH_0^{\Psi_\pm}}=x+2t\Psi'_\pm\left(|D|^2\right)D.\label{Lem1_2}
\end{equation}
Therefore, its time evolution for the monotonically increasing case is estimated to be
\begin{eqnarray}
\left\|xe^{-{\rm i}tH_0^{\Psi_+}}\phi\right\|^2\Eqn{=}\left\|\left(x+2t\Psi'_+\left(|D|^2\right)D\right)f(D)\phi\right\|^2\nonumber\\
\Eqn{\leqslant}2\left\|x\phi\right\|^2+8nt^2\Psi'_+(R^2)^2R^2\|\phi\|^2.
\end{eqnarray}
Take $\Gamma_+\in\mathbb{R}$ such that
\begin{equation}
\Gamma_+\geqslant\max\left\{4\sqrt{n}\Psi'_+(R^2)R, 1\right\},
\end{equation}
then the estimate of $e^{-{\rm i}tH_0^{\Psi_+}}\phi$ outside a sphere of radius $\Gamma_+t$ is
\begin{eqnarray}
\lefteqn{\int_{|x|>\Gamma_+t}\left|\left(e^{-{\rm i}tH_0^{\Psi_+}}\phi\right)(x)\right|^2dx\leqslant\int_{|x|>\Gamma_+t}\frac{|x|^2}{\Gamma_+^2t^2}\left|\left(e^{-{\rm i}tH_0^{\Psi_+}}\phi\right)(x)\right|^2dx}\nonumber\\
\Eqn{}\leqslant\frac{1}{\Gamma_+^2t^2}\left\|xe^{-{\rm i}tH_0^{\Psi_+}}\phi\right\|^2\leqslant\frac{1}{\Gamma_+^2t^2}\left\{2\left\|x\phi\right\|^2+8nt^2\Psi'_+(R^2)^2R^2\|\phi\|^2\right\}\nonumber\\
\Eqn{}\leqslant\frac{2}{\Gamma_+^2t^2}\left\|x\phi\right\|^2+\frac{1}{2}\|\phi\|^2.\label{Lem1_3}
\end{eqnarray}
We write $(\cdot,\cdot)_{L^2(\mathbb{R}^n)}=(\cdot,\cdot)$ and compute
\begin{eqnarray}
\lefteqn{\frac{1}{\kappa}\left(V^{\rm L}e^{-{\rm i}tH_0^{\Psi_+}}\phi,e^{-{\rm i}tH_0^{\Psi_+}}\phi\right)=\int_{|x|>\Gamma_+t}+\int_{|x|\leqslant\Gamma_+t}\langle x\rangle^{-\gamma_{\rm L}}\left|\left(e^{-{\rm i}tH_0^{\Psi_+}}\phi\right)(x)\right|^2dx}\nonumber\\
\Eqn{}\geqslant\int_{|x|\leqslant\Gamma_+t}\langle x\rangle^{-\gamma_{\rm L}}\left|\left(e^{-{\rm i}tH_0^{\Psi_+}}\phi\right)(x)\right|^2dx\nonumber\\
\Eqn{}\geqslant\langle \Gamma_+t\rangle^{-\gamma_{\rm L}}\int_{|x|\leqslant\Gamma_+t}\left|\left(e^{-{\rm i}tH_0^{\Psi_+}}\phi\right)(x)\right|^2dx\nonumber\\
\Eqn{}=\langle \Gamma_+t\rangle^{-\gamma_{\rm L}}\|\phi\|^2-\langle \Gamma_+t\rangle^{-\gamma_{\rm L}}\int_{|x|>\Gamma_+t}\left|\left(e^{-{\rm i}tH_0^{\Psi_+}}\phi\right)(x)\right|^2dx.\hspace{25mm}
\end{eqnarray}
Using the inequality \eqref{Lem1_3}, we have
\begin{gather}
\frac{1}{\kappa}\left(V^{\rm L}e^{-{\rm i}tH_0^{\Psi_+}}\phi,e^{-{\rm i}tH_0^{\Psi_+}}\phi\right)\geqslant\frac{1}{2}\langle \Gamma_+t\rangle^{-\gamma_{\rm L}}\|\phi\|^2-\langle \Gamma_+t\rangle^{-\gamma_{\rm L}}\times\frac{2}{\Gamma_+^2t^2}\left\|x\phi\right\|^2\nonumber\\
\geqslant c_{+,1}t^{-\gamma_{\rm L}}\|\phi\|^2-c_{+,2}t^{-2-\gamma_{\rm L}}\left\|x\phi\right\|^2.\label{Lem1_4}
\end{gather}
For the last inequality in \eqref{Lem1_4}, we set $c_{+,1}$ and $c_{+,2}$ using
\begin{eqnarray}
\frac{1}{2}\langle \Gamma_+t\rangle^{-\gamma_{\rm L}}\Eqn{\geqslant}\frac{1}{2}\left(1+\Gamma_+t\right)^{-\gamma_{\rm L}}\geqslant\frac{1}{2}\left(2\Gamma_+\right)^{-\gamma_{\rm L}}t^{-\gamma_{\rm L}}=c_{+,1}t^{-\gamma_{\rm L}},\\
2\frac{\langle \Gamma_+t\rangle^{-\gamma_{\rm L}}}{\Gamma_+^2t^2}\Eqn{\leqslant}2\frac{\left(\Gamma_+t\right)^{-\gamma_{\rm L}}}{\Gamma_+^2t^2}=2\Gamma_+^{-2-\gamma_{\rm L}}t^{-2-\gamma_{\rm L}}=c_{+,2}t^{-2-\gamma_{\rm L}}.
\end{eqnarray}
In contrast, for the monotonically decreasing case, we have
\begin{equation}
\left\|xe^{-{\rm i}tH_0^{\Psi_-}}\phi\right\|^2\leqslant2\left\|x\phi\right\|^2+8nt^2\Psi'_-(\epsilon^2)^2\epsilon^2\|\phi\|^2.\label{Lem1_5}
\end{equation}
Therefore,
\begin{equation}
\frac{1}{\kappa}\left(V^{\rm L}e^{-{\rm i}tH_0^{\Psi_-}}\phi,e^{-{\rm i}tH_0^{\Psi_-}}\phi\right)\geqslant c_{-,1}t^{-\gamma_{\rm L}}\|\phi\|^2-c_{-,2}t^{-2-\gamma_{\rm L}}\left\|x\phi\right\|^2.
\end{equation}
holds for $\Gamma_-\in\mathbb{R}$ which satisfies
\begin{equation}
\Gamma_-\geqslant\max\left\{4\sqrt{n}\Psi'_-(\epsilon^2)\epsilon, 1\right\}.
\end{equation}
\end{proof}

\begin{Lem}\label{Lem2}
For $t>0$ and $N\in\mathbb{N}$, there exist positive constants $c_{\pm,3}$ and $c_{\pm,4}$ such that
\begin{equation}
\left\|V^{\rm L}e^{-{\rm i}tH_0^{\Psi_\pm}}\phi\right\|_{L^2(\mathbb{R}^n)}\leqslant c_{\pm,3}t^{-\gamma_{\rm L}}\|\phi\|_{L^2(\mathbb{R}^n)}+c_{\pm,4}t^{-N}\left\|\langle x\rangle^N\phi\right\|_{L^2(\mathbb{R}^n)}\label{Lem2_1}
\end{equation}
\end{Lem}

\begin{proof}
This proof follows in almost the same way as the proof for the existence of the wave operators (see \eqref{thm1_2}). For the monotonically increasing case, \eqref{Lem2_1} follows by setting $c_{+,3}=|\kappa|\left(\Psi'_-(\epsilon^2)\epsilon\right)^{-\gamma_{\rm L}}$ and $c_{+,4}=|\kappa|C_{+,N,\epsilon, R}$. For the monotonically decreasing case, \eqref{Lem2_1} follows by $c_{-,3}=|\kappa|\left(\Psi'_-(R^2\right)R)^{-\gamma_{\rm L}}$ and $c_{-,4}=|\kappa|C_{-,N,\epsilon, R}$. 
\end{proof}

We have now gathered everything required to prove the non-existence of the wave operators.

\begin{proof}[Proof of the nonexistence of the wave operators]
We assume that
\begin{equation}
\slim_{t\rightarrow\infty}e^{{\rm i}tH_{\rm L}^{\Psi_\pm}}e^{-{\rm i}tH_0^{\Psi_\pm}}
\end{equation}
exists and put
\begin{equation}
\phi_\pm=\lim_{t\rightarrow\infty}e^{{\rm i}tH_{\rm L}^{\Psi_\pm}}e^{-{\rm i}tH_0^{\Psi_\pm}}\phi\in L^2(\mathbb{R}^n).
\end{equation}
There exist $T_\pm>0$ such that
\begin{equation}
\left\|e^{{\rm i}tH_{\rm L}^{\Psi_\pm}}e^{-{\rm i}tH_0^{\Psi_\pm}}\phi-\phi_\pm\right\|\leqslant\frac{|\kappa|c_{\pm,1}}{2c_{\pm,3}}\|\phi\|
\end{equation}
for all $t\geqslant T_\pm$. We take $t_1$ and $t_2$ such that $t_2\geqslant t_1\geqslant\max\{T_\pm, 1\}$ and compute
\begin{eqnarray}
\lefteqn{\left|\left(\left\{e^{{\rm i}t_2H_{\rm L}^{\Psi_\pm}}e^{-{\rm i}t_2H_0^{\Psi_\pm}}-e^{{\rm i}t_1H_{\rm L}^{\Psi_\pm}}e^{-{\rm i}t_1H_0^{\Psi_\pm}}\right\}\phi,\phi_\pm\right) \right|}\nonumber\\
\Eqn{}=\left|\int_{t_1}^{t_2}\frac{d}{dt}\left(e^{{\rm i}tH_{\rm L}^{\Psi_\pm}}e^{-{\rm i}tH_0^{\Psi_\pm}}\phi,\phi_\pm\right)dt\right|=\left|\int_{t_1}^{t_2}\left(V^{\rm L}e^{-{\rm i}tH_0^{\Psi_\pm}}\phi,e^{-{\rm i}tH_{\rm L}^{\Psi_\pm}}\phi_\pm\right)dt\right|\nonumber\\
\Eqn{}=\left|\int_{t_1}^{t_2}\left(V^{\rm L}e^{-{\rm i}tH_0^{\Psi_\pm}}\phi,e^{-{\rm i}tH_0^{\Psi_\pm}}\phi\right)dt\right.\nonumber\\
\Eqn{}\hspace{30mm}+\left.\int_{t_1}^{t_2}\left(V^{\rm L}e^{-{\rm i}tH_0^{\Psi_\pm}}\phi,\left\{e^{-{\rm i}tH_{\rm L}^{\Psi_\pm}}\phi_\pm-e^{-{\rm i}tH_0^{\Psi_\pm}}\phi\right\}\right)dt\right|\nonumber\\
\Eqn{}\geqslant|\kappa|\int_{t_1}^{t_2}\frac{1}{\kappa}\left(V^{\rm L}e^{-{\rm i}tH_0^{\Psi_\pm}}\phi,e^{-{\rm i}tH_0^{\Psi_\pm}}\phi\right)dt\nonumber\\
\Eqn{}\hspace{30mm}-\int_{t_1}^{t_2}\left\|V^{\rm L}e^{-{\rm i}tH_0^{\Psi_\pm}}\phi\right\|\left\|e^{-{\rm i}tH_{\rm L}^{\Psi_\pm}}\phi_\pm-e^{-{\rm i}tH_0^{\Psi_\pm}}\phi\right\|dt.
\end{eqnarray}
By virtue of Lemmas \ref{Lem1} and \ref{Lem2}, we conclude that
\begin{eqnarray}
\lefteqn{2\|\phi\|\left\|\phi_\pm\right\|\geqslant\frac{|\kappa|c_{\pm,1}}{2}\|\phi\|^2\int_{t_1}^{t_2}t^{-\gamma_{\rm L}}dt}\nonumber\\
\Eqn{}-|\kappa|c_{\pm,2}\left\|x\phi\right\|^2\int_{t_1}^{t_2}t^{-2-\gamma_{\rm L}}dt-\frac{|\kappa|c_{\pm,1}c_{\pm,4}}{2c_{\pm,3}}\|\phi\|\left\|\langle x\rangle^N\phi\right\|\int_{t_1}^{t_2}t^{-N}dt\nonumber\\
\Eqn{}\longrightarrow\infty
\end{eqnarray}
as $t_2$ goes to infinity because $\gamma_{\rm L}\leqslant1$, and we can choose $N\in\mathbb{N}$ as $N\geqslant2$. This leads to a contradiction.
\end{proof}
\noindent\textbf{Acknowledgments.} The first author was partially supported by the Grant-in-Aid for Young Scientists (B) \#16K17633 and Scientific Research (C) \#20K03625 from JSPS.


\end{document}